%% file: costly.tex
\documentclass[10pt,conference]{IEEEtran} 
\IEEEoverridecommandlockouts
\usepackage{cite}
\usepackage{algorithmic}
\usepackage{graphicx}
\usepackage{textcomp}
\def\BibTeX{{\rm B\kern-.05em{\sc i\kern-.025em b}\kern-.08em
    T\kern-.1667em\lower.7ex\hbox{E}\kern-.125emX}}

\usepackage{textcase}
\usepackage{thmtools,thm-restate}
\usepackage{hyperref}
\usepackage{booktabs} 
\usepackage{mathtools}
\usepackage{amsmath,amssymb,mathrsfs,amsfonts,dsfont,amsthm,soul}
\usepackage{enumerate}
\usepackage[english]{babel}
\usepackage{comment}
\usepackage{epsfig,amsfonts,latexsym,amssymb,shadow,amsmath,wrapfig}
\usepackage[margin=10pt,font=small,labelfont=bf,justification=justified]{caption}

\DeclareMathOperator*{\argmax}{\arg\!\max}
\newcommand{\minitab}{\hspace*{2em}}

\newtheorem{theorem}{Theorem}

\newtheorem{lemma}{Lemma}

\newtheorem{assum}{Assumption}

\begin{document}

\title{Optimal Mechanism Design with Flexible Consumers and Costly Supply 
}

\author{\IEEEauthorblockN{1\textsuperscript{st} Shiva Navabi}
\IEEEauthorblockA{\textit{Electrical Engineering Department} \\
\textit{University of Southern California}\\
Los Angeles, USA \\
navabiso@usc.edu}
\and
\IEEEauthorblockN{2\textsuperscript{nd} Ashutosh Nayyar}
\IEEEauthorblockA{\textit{Electrical Engineering Department} \\
\textit{University of Southern California}\\
Los Angeles, USA \\
ashutosn@usc.edu}
}

\maketitle
\vspace{-10em}
\begin{abstract}
The problem of designing a profit-maximizing, Bayesian incentive compatible and individually rational mechanism with  flexible consumers and costly heterogeneous supply  is considered. In our setup, each consumer is associated with a  \emph{flexibility set} that describes the subset of goods the consumer is equally interested in. Each consumer wants to consume \emph{one good from its flexibility set}. The flexibility set of a consumer  and the utility it gets from consuming a  good from its flexibility set are  its private information. We adopt the flexibility model of \cite{navabi2016optimal} and focus on the case of nested flexibility sets --- each consumer's flexibility set can be one of   $k$ nested sets. Examples of settings with this inherent nested structure are provided. On the supply side, we assume that the seller has an initial stock of free supply but it can purchase more goods for each of the nested sets at fixed exogenous prices. We characterize the allocation and purchase rules for a profit-maximizing, Bayesian incentive compatible and individually rational mechanism as the solution to an integer program. The optimal payment function is pinned down by the optimal allocation  rule in the form of an integral equation. We show that  the nestedness of flexibility sets can be exploited to obtain a simple description of the optimal allocations, purchases and payments in terms of  thresholds that can be computed through a straightforward iterative procedure.  
\end{abstract}




\input{intro}

\subsection{Organization} \label{sec:ORG}
The rest of the paper is organized as follows: we discuss the problem formulation and the mechanism setup in Section \ref{sec:Formul}. In Section \ref{sec:BIC-IR-Mechs}, we characterize incentive compatibility and individual rationality constraints for the mechanism.  In Section \ref{sec:RMM} we characterize the allocation and purchase rules for a profit-maximizing, Bayesian incentive
compatible and individually rational mechanism as the solution to an integer program. Borrowing the results in \cite{navabi2016optimal}, we simplify the optimal allocations, purchases and payments and characterize them in terms of simple thresholds. We conclude the paper in Section \ref{sec:recap} with a summary of our findings and point out potential directions for further research. 


\subsection{Notations}

$\{ 0,1\}^{N \times M}$ denotes the space of $N\times M$ dimensional matrices with entries that are either 0 or 1. $\mathbb{Z}_+$ is the set of non-negative integers. For a set $\mathcal{A}$, $|\mathcal{A}|$ denotes the cardinality of $\mathcal{A}$.  $x^+$ is the positive part of the real number $x$, that is, $x^+ = \max(x,0)$. 
Vector inequalities are component-wise; that is, for two $1\times n$ dimensional vectors $\textbf{u}=(u_1, \cdots, u_n)$ and $\textbf{v}=(v_1, \cdots, v_n)$, $\textbf{u} \le \textbf{v}$ implies that $u_i \le v_i \; , \text{for} \; i=1, \cdots, n$.  $\mathds{1}_{\{a\le b\}}$ denotes 1 if the inequality in the subscript is true and 0 otherwise. $\mathbb{E}$ denotes the expectation operator. For a random variable $\theta$, $\mathbb{E}_{\theta}$ denotes that the expectation is with respect to the probability distribution of $\theta$. 

\section{Problem Formulation}\label{sec:Formul}
We consider a setup where a seller has $M$ goods and $N$ potential consumers. $\mathcal{M} = \{ 1, 2, \cdots, M\}$ denotes the set of goods and $\mathcal{N} = \{1, 2, \cdots, N\}$ denotes the set of potential consumers. Consumer $i$, $i \in \mathcal{N}$, has a flexibility set $\phi_i \subset \mathcal{M}$ which represents the set of goods the consumer is equally interested in. Consumer $i$ can consume at most one good from its flexibility set $\phi_i$. We assume that the flexibility set of each consumer can be one of $k$ nested sets. That is, we have $k$ nested subsets of the set of goods:
\begin{equation}\label{Bsets}
\mathcal{B}_1 \subset \mathcal{B}_2 \subset \cdots \subset \mathcal{B}_k \subseteq \mathcal{M},
\end{equation}  
and $\phi_i \in \{ \mathcal{B}_1, \mathcal{B}_2, \cdots, \mathcal{B}_k\}$ for every $i \in \mathcal{N}$. We also define $\mathcal{B}_0$ as the empty set, i.e., $\mathcal{B}_0 = \emptyset$. Based on their flexibility sets, we can divide the consumers into $k$ classes: $C_l$ is the set of consumers with flexibility set $\mathcal{B}_l$. Clearly, $\mathcal{N} = \bigcup\limits_{i=1}^k \mathcal{C}_i$ and $\mathcal{C}_i \cap \mathcal{C}_j = \emptyset \; , \; \forall i \neq j$. We define \vspace{-.6em}
\begin{equation}  \label{mi}
\begin{split}
 n_l := |\mathcal{C}_l| \; , ~ \;
m_l := |\mathcal{B}_l \: \backslash \: \mathcal{B}_{l-1}| ,  \; l =1, 2, \ldots, k .
\end{split}
\end{equation}
We also define the vectors \textbf{n} and \textbf{m} as \vspace{-.6em}
\begin{equation} \label{SupplyDemandProfs}
\begin{split}
\textbf{n} := (n_1, n_2, \cdots, n_k) \;  , \; \; 
\textbf{m} := (m_1, m_2, \cdots, m_k).
\end{split}
\end{equation}
The vector \textbf{n} is referred to as the \emph{demand profile} and the vector \textbf{m} is referred to as the \emph{supply profile}.

If $\phi_i = \mathcal{B}_j$, we say that consumer $i$'s \textit{flexibility level}, denoted by $b_i$, is $j$. Consumer $i$'s utility for receiving a good from $\phi_i$ is $\theta_i$. We assume that $\theta_i$ and $b_i$ are consumer $i$'s private information and are unknown to other consumers as well as the seller. We assume that $(\theta_i, b_i), i \in \mathcal{N}$, are independent random pairs taking values in the product set $[\theta_i^{max}, \theta_i^{min}] \times \{1, 2, \cdots, k\}$\footnote{$\theta_i^{min}$ is assumed to be non-negative for all $i \in \mathcal{N}$.}, $i \in \mathcal{N}$. The probability distributions $f_i$ of $(\theta_i, b_i), i \in \mathcal{N}$, are assumed to be common knowledge. We define $\theta \coloneqq (\theta_1, \theta_2, \cdots, \theta_N)$ and $b \coloneqq (b_1, b_2, \cdots, b_N)$ as the consumers' valuations profile and flexibility levels profile, respectively. $f(\theta, b)$ is the joint probability distribution of $(\theta, b)$. Let $\Theta_i = [\theta_i^{min}, \theta_i^{max}]$ and $\Theta \coloneqq \prod\limits_{i=1}^N \Theta_i$. The pair $(\theta_i, b_i)$ is referred to as consumer $i$'s type. 

While the initial set of goods ($\mathcal{M}$) is available to the seller at zero cost, it can purchase more goods for addition to the sets $\mathcal{B}_i \setminus \mathcal{B}_{i-1}, i = 1, \cdots,k$. Each additional good for the sets $\mathcal{B}_i \setminus \mathcal{B}_{i-1}, i = 1, \cdots,k$ can be purchased at the price of $p_i, i = 1, \cdots, k$, where $p_1 > p_2 > \cdots > p_k$. 

\subsection{Direct Mechanisms} \label{sec:direct}
We consider direct mechanisms where, for each $i \in \mathcal{N}$, consumer $i$ reports a valuation from the set $\Theta_i$ and a flexibility level from the set $\{1, 2, \cdots, k\}$ to the seller. The consumers can misreport their valuations as well as their flexibility levels. We assume however that consumers cannot over-report their flexibility levels:

\begin{assum}\label{assum:feasb}
For each $i \in \mathcal{N}$, consumer $i$'s reported flexibility level $c_i$ cannot exceed its true flexibility level $b_i$.
\end{assum}
The above assumption can be justified by noting that consumers gain no utility from getting a good outside their true flexibility set and may in fact suffer a significant disutility if allocated a good outside their true flexibility set. To avoid the risk of getting an unusable or damaging good, consumers may reasonably restrict themselves to under-reporting or truthfully reporting their flexibility levels.   

A mechanism consists of an allocation rule $\xi$, a payment rule $t$ and a purchase decision rule $g$. The allocation rule $\xi$ is a function from the type profile space $\Theta \times \{1, 2, \cdots, k\}^N$ to $\{0,1\}^N$ with its $i$th entry $\xi_i$ being equal to $1$ if consumer $i$ receives a good and $0$ otherwise. We assume that the allocation rule $\xi$ does not give a consumer any good that is outside its \textit{reported} flexibility set.  This can be formalized as follows: 
\begin{assum}\label{assum:Nonz}
We assume that for each $i \in \mathcal{N}$, if consumer $i$ reports flexibility level $c_i$, the mechanism either allocates a good from $\mathcal{B}_{c_i}$ to consumer $i$ or it does not allocate any good to consumer $i$. 
\end{assum}
The above assumption simply means that the mechanism respects the consumer's reported flexibility constraint. 

The purchase decision rule $g$ is a function from $\Theta \times \{1, 2, \cdots, k\}^N$ to $\mathbb{Z}_+^{k}$ with its $i$th entry $g_i$ being the number of goods additionally purchased at the price of $p_i$ and added to the set $\mathcal{B}_i \setminus \mathcal{B}_{i-1}$.

 We require that each of the ($M$ available and the additionally purchased) goods be allocated to at most one consumer and that each consumer receives at most one good. The allocation and purchase decision vector pairs $(\xi,g)$ that satisfy these two constraints together, are called feasible decision vector pairs. The feasibility constraints are formally stated in the following lemma.
 
 \begin{lemma} \label{lemma:Adeq}
 For a reported type profile $(r,c) \in \Theta \times \{ 1, 2, \cdots, k \}^N$, the allocation $\xi(r,c)$ and purchase $g(r,c)$ will satisfy the feasibility constraints if and only if
 \begin{equation} \label{feas}
 \sum\limits_{l: c_l \le i} \xi_l(r,c) \le \sum\limits_{j \le i} (m_j + g_j(r,c)) \; , \; i = 1, 2, \cdots, k.
 \end{equation}
 \begin{proof}
 The proof is based on the arguments in \cite[Section V-A]{navabi2016optimal}.
 \end{proof}
 \end{lemma} 
 \vspace{-.7em}
 We define $\mathcal{S}(c) \subset \{0,1\}^N \times \mathbb{Z}_+^k$ as the set of all $\xi(r,c)$ and $g(r,c)$ that satisfy \eqref{feas}. That is,
\begin{equation}\label{Sset}
\begin{split}
\mathcal{S}(c) \coloneqq \Big\{&\Big(\bold{x} \in \{0,1\}^N, \bold{y} \in \mathbb{Z}_+^{k}\Big) : \\
&\sum\limits_{l : c_l \le i} x_l \le \sum\limits_{j \le i} (m_j + y_j) \; , \; i = 1, 2, \cdots , k \Big\}.
\end{split}
\end{equation}

The payment rule is a mapping from $\Theta \times \{1, 2, \cdots, k\}^N$ to $\mathbb{R}^N$ with the $i$th component $t_i$ being the payment charged to consumer $i$.

Consider a mechanism $(\xi, g, t)$ and suppose consumers report valuations $r \coloneqq (r_1, r_2, \cdots, r_N)$ and flexibility levels $c \coloneqq (c_1, c_2, \cdots, c_N)$\footnote{consumers may not report their valuations and/or flexibility levels truthfully, so $r_i$ and $c_i$ may be different from $\theta_i$ and $b_i$, respectively.}. The mechanism then results in an allocation vector $\xi(r, c)$, purchase decision vector $g(r, c)$ and payments $t(r, c)$. Consumer $i$'s utility function can then be written in terms of its true valuation $\theta_i$, true flexibility level $b_i$, the reported valuations $r$ and the reported flexibility levels $c$ as \vspace{-1em}
\begin{align}
u_i(\theta_i, r, b_i, c) = \theta_i \xi_i(r, c) - t_i(r,c).
\end{align} 
The seller's objective is to find a mechanism that maximizes its expected profit while satisfying \textit{Bayesian Incentive Compatibility} and \textit{Individual Rationality} constraints. We describe these constraints below.

In a Bayesian incentive compatible (BIC) mechanism, truthful reporting of private information (valuations and flexibility levels in our setup) constitutes an equilibrium of the Bayesian game induced by the mechanism. In other words, each consumer would prefer to report its true valuation and flexibility level provided that all other consumers have adopted truth-telling strategy. Bayesian incentive compatibility can be
described by the following constraint:
\begin{equation}
\begin{split}
&\mathbb{E}_{\theta_{-i}, b_{-i}}\Big[ \theta_i \xi_i(\theta,b) - t_i(\theta,b) \Big] \: \ge  \\
&\mathbb{E}_{\theta_{-i}, b_{-i}}\Big[ \theta_i \xi_i(r_i, \theta_{-i}, c_i, b_{-i}) - t_i(r_i, \theta_{-i}, c_i, b_{-i}) \Big], \\
&\forall \theta_i , r_i \in \Theta_i  \; , c_i \le b_i, \; c_i , b_i \in \{1, 2, \cdots, k\} \; , \; \forall i \in \mathcal{N}. 
\end{split}
\raisetag{3.4\baselineskip}
\label{BIC}
\end{equation}
\eqref{BIC} states that if consumer $i$ with type $(\theta_i, b_i)$ reports some other type $(r_i, c_i)$, its expected utility will be \emph{no better than} the expected utility it gets if it reports its type truthfully. Note that Assumption \ref{assum:feasb} implies that the BIC constraint in \eqref{BIC} need not consider the case of $c_i > b_i$.


Individual Rationality (IR) constraint implies that the consumer's expected utility at the truthful reporting equilibrium is non-negative. This can be expressed as:
\begin{equation} \label{IR}
\begin{split}
&\mathbb{E}_{\theta_{-i}, b_{-i}}\Big[ \theta_i \xi_i(\theta,b) - t_i(\theta,b) \Big] \; \ge \; 0  \; \; , \\
&\forall \theta_i \in \Theta_i \; , b_i \in \{1, 2, \cdots, k\} \; , \; \forall i \in \mathcal{N}.
\end{split}
\end{equation}

The expected profit under a BIC and IR mechanism is $\mathbb{E}_{\theta, b}\Big\{\sum\limits_{i=1}^N t_i(\theta, b) - \sum\limits_{j=1}^k p_j g_j(\theta, b) \Big\}$ when all consumers adopt the truthful strategy.

The mechanism design problem  can now be formulated as
\begin{equation}
\begin{split}
\max\limits_{(\xi,g,t)} \; \; \; &\mathbb{E}_{\theta, b}\Big\{\sum\limits_{i=1}^N t_i(\theta, b) - \sum\limits_{j=1}^k p_j g_j(\theta, b)\Big\} \;  , \\
\hspace{-.5em}\text{subject to} \; \; &(\xi(\theta, b), g(\theta, b)) \in \mathcal{S}(b), \forall (\theta, b) \in \Theta \times \{1, 2, \cdots, k\}^N, \\
&\text{\eqref{BIC}, \eqref{IR}}.
\end{split}
\raisetag{1\baselineskip}
\end{equation}

\section{Characterization of BIC and IR Mechanisms}  \label{sec:BIC-IR-Mechs}
Suppose all consumers other than $i$ report their valuations and flexibility levels truthfully. 
We can then define consumer $i$'s expected allocation and payment under the mechanism $(\xi,g,t)$ when it reports $r_i \in \Theta_i\; , \; c_i \in \{1, 2, \cdots, k\}$ as:
\begin{equation} \label{Qi}
\Xi_i(r_i, c_i) \coloneqq \mathbb{E}_{\theta_{-i}, b_{-i}}\Big[ \xi_i(r_i, \theta_{-i}, c_i, b_{-i}) \Big], 
\end{equation}
\vspace{-7pt}
\begin{equation} \label{Ti}
T_i(r_i, c_i) \coloneqq \mathbb{E}_{\theta_{-i}, b_{-i}}\Big[ t_i(r_i, \theta_{-i}, c_i, b_{-i}) \Big].
\end{equation}
We can now rewrite equations \eqref{BIC} and \eqref{IR} in terms of the interim quantities defined in \eqref{Qi}-\eqref{Ti}. 
The BIC constraint for misreports of valuations and flexibility levels becomes:
\begin{equation} \label{BIC_i_QT}
\begin{split}
&\theta_i \Xi_i(\theta_i, b_i) - T_i(\theta_i, b_i) \ge \theta_i \Xi_i(r_i, c_i) - T_i(r_i, c_i) \;  , \\
&\forall \theta_i , r_i \in \Theta_i \; , \; c_i \le b_i \; , \; c_i , b_i \in \{1, 2, \cdots, k\} \; , \; \forall i \in \mathcal{N}. 
\end{split}
\raisetag{2\baselineskip}
\end{equation}
The IR constraint is rewritten as:
\begin{equation}\label{IR_i_QT}
\begin{split}
&\theta_i \Xi_i(\theta_i, b_i) - T_i(\theta_i, b_i) \ge 0 \; , \\
&\forall \theta_i \in \Theta_i \; , b_i \in \{1, 2, \cdots, k\} \; , \; \forall i \in \mathcal{N}.
\end{split}
\end{equation}
The  BIC constraint in \eqref{BIC_i_QT} captures all possible ways that a consumer may misreport its private information. It includes the following two special sub-classes of constraints:
\begin{enumerate}
\item BIC constraint for misreporting  only valuation:
\begin{equation} \label{BICTheta_i_QT}
\begin{split}
&\theta_i \Xi_i(\theta_i, b_i) - T_i(\theta_i, b_i) \ge \theta_i \Xi_i(r_i, b_i) - T_i(r_i, b_i) \;  , \\
&\forall \theta_i , r_i \in \Theta_i \; , \; \forall b_i \in \{1, 2, \cdots, k\} \; , \; \forall i \in \mathcal{N}. 
\end{split}
\raisetag{1\baselineskip}
\end{equation} 
\item BIC constraint for misreporting only  flexibility level:
\begin{equation} \label{BICb_i_QT}
\begin{split}
&\hspace{-1em}\theta_i \Xi_i(\theta_i, b_i) - T_i(\theta_i, b_i) \ge \theta_i \Xi_i(\theta_i, c_i) - T_i(\theta_i, c_i) \;  , \\
&\hspace{-1em}\forall \theta_i \in \Theta_i \; , \; c_i \le b_i \; , \; c_i , b_i \in \{1, 2, \cdots, k\} \; , \; \forall i \in \mathcal{N}. 
\end{split}
\raisetag{1\baselineskip}
\end{equation}
\end{enumerate}
The following result relates the above constraints for ``one-dimensional'' misreports to the general BIC constraint in \eqref{BIC_i_QT}.
\begin{lemma}\label{lemma:BIC2D} The BIC constraint for misreporting both valuation and flexibility level  implies and is implied by the BIC constraints for misreporting only valuation and misreporting only flexibility level. That is, \eqref{BIC_i_QT} holds if and only if \eqref{BICTheta_i_QT} and \eqref{BICb_i_QT} hold.
\begin{proof} 
The proof is similar to  the proof for Lemma 1 in \cite{navabi2016optimal}. 
\end{proof}
\end{lemma}
Lemma \ref{lemma:BIC2D} allows us to replace the general  BIC constraint for two-dimensional misreports by the simpler one-dimensional BIC constraints given in \eqref{BICTheta_i_QT} and \eqref{BICb_i_QT}.   The mechanism design problem now becomes: \vspace{-.8em}
\begin{equation*} 
\begin{split}
\max\limits_{(\xi,g,t)} \; \; \; &\mathbb{E}_{\theta, b}\Big\{\sum\limits_{i=1}^N t_i(\theta, b) - \sum\limits_{j=1}^k p_j g_j(\theta, b)\Big\} \;  , \\
\text{subject to} \; \; &(\xi(\theta, b), g(\theta, b)) \in \mathcal{S}(b), \; \; \\
&\forall (\theta, b) \in \Theta \times \{1, 2, \cdots, k\}^N, 
\text{\eqref{IR_i_QT}, \eqref{BICTheta_i_QT}, \eqref{BICb_i_QT}}.
\end{split}
\end{equation*}

We will now derive alternative characterizations of the constraints \eqref{IR_i_QT}, \eqref{BICTheta_i_QT}, \eqref{BICb_i_QT} that will be helpful for finding the optimal mechanism.
 
\begin{lemma}\label{thm:one}
 A mechanism $(\xi,g,t)$ satisfies the BIC constraint for misreporting only valuation (as given in \eqref{BICTheta_i_QT}) if and only if for all $i \in \mathcal{N}$, $\Xi_i(r_i, b_i)$ is non-decreasing in $r_i$ and  \vspace{-5pt}
\begin{align}
\hspace{-1em}T_i(r_i, b_i) = K_i(b_i) + r_i \Xi_i(r_i, b_i) - \int\limits_{\theta_i^{\text{min}}}^{r_i} \Xi_i(s, b_i) \: ds, \; \forall b_i. 
\label{Thrm1}
\end{align}
\end{lemma}
\begin{proof}
The proof is similar to the arguments in chapters 2-3 of \cite{borgers2015introduction} as well as the proof for Lemma 2 in \cite{navabi2016optimal} for characterizing BIC mechanisms. 
\end{proof}

\begin{lemma}\label{lemma:ir}
Suppose the mechanism $(\xi,g,t)$ satisfies the BIC constraint for misreporting only valuation (as given in \eqref{BICTheta_i_QT}). Then, it satisfies the IR constraint \eqref{IR_i_QT} if and only if 
\begin{equation} \label{Ki}
\theta_i^{\text{min}} \; \Xi_i(\theta_i^{\text{min}}, b_i) - T_i(\theta_i^{\text{min}}, b_i)   \ge 0.
\end{equation}
\end{lemma}
\begin{proof}
Clearly \eqref{IR_i_QT} implies \eqref{Ki}. The converse follows from Lemma \ref{thm:one} by noting that \[K_i(b_i) = T_i(\theta_i^{min}, b_i) - \theta_i^{min} \Xi_i(\theta_i^{min}, b_i),\]
and that the right hand side above is non-positive due to \eqref{Ki}.
\end{proof}

Using the  above two lemmas, we  derive a sufficient condition for the mechanism to satisfy the BIC constraint for misreporting only flexibility level.
\begin{lemma}\label{lem:BICSuffb} Suppose the mechanism $(\xi,g,t)$ is individually rational and satisfies the BIC constraint for misreporting only valuation (as given in \eqref{BICTheta_i_QT}). Then the mechanism $(\xi,g,t)$ satisfies the BIC constraint for misreporting only flexibility level if the following are true:
\begin{enumerate}[(i)]
\item $\Xi_i(\theta_i,c_i)$ is non-decreasing in $c_i \; , \; \forall \theta_i \in \Theta_i, \forall i \in \mathcal{N}$, and
\item $T_i(\theta_i^{\text{min}}, c_i) = 0 \; , \; \forall c_i \in \{ 1, 2, \cdots, k \}\; , \; \forall i \in \mathcal{N}$. 
\end{enumerate}
\begin{proof} 
The proof is similar to the proof for Lemma 4 in \cite{navabi2016optimal}. 
\end{proof}
\end{lemma}

\section{profit Maximizing Mechanism}\label{sec:RMM}
We can now use the results of Section \ref{sec:BIC-IR-Mechs} to simplify the objective of the mechanism design problem. We define  
\begin{equation}
w_i(\theta_i, b_i) := \Big( \theta_i - \frac{1 - F_{i}(\theta_i | b_i)}{f_{i}(\theta_i | b_i)}\Big),
\end{equation}
 where  $f_{i}(\theta_i | b_i)$ is the conditional probability density function of consumer $i$'s valuation conditioned on its flexibility level $b_i$ and $F_{i}(\theta_i | b_i)$ is the corresponding cumulative distribution function. $w_i(\theta_i, b_i)$  is referred to as consumer $i$'s \textit{virtual type or virtual valuation} in economics terminology \cite{borgers2015introduction}. 
 
 For a mechanism that is individually rational and Bayesian incentive compatible, we can use the result in Lemma \ref{thm:one} to plug in the expression for  $T_i(\theta_i, b_i)$. After some simplifications we obtain  
 \begin{align} \label{eq:ExpRev}
&\mathbb{E}_{\theta, b}\Big\{\sum\limits_{i=1}^N t_i(\theta, b) - \sum\limits_{j=1}^k p_j g_j(\theta, b)\Big\}  \\
&= \mathbb{E}_{b_i}\Big[ K_i(b_i) \Big]   \notag\\
&+ \sum\limits_{b}\int_{\theta}\Big[ \xi_i(\theta, b) w_i(\theta_i, b_i) - \sum\limits_{j=1}^k p_j g_j(\theta, b) \Big] f(\theta, b) d\theta.  \notag
\end{align}
The second term on the right hand side in \ref{eq:ExpRev} is completely determined by the choices of the allocation $\xi(\cdot, \cdot)$ and purchase $g(\cdot, \cdot)$ functions. 
Also, note that  Lemmas \ref{thm:one} and \ref{lemma:ir} imply that $K_i(b_i) = T_i(\theta_i^{\text{min}}, b_i) - \theta^{min}_i\Xi_i(\theta_i^{\text{min}}, b_i)  \le 0$. Therefore, a mechanism $(\xi,g,t)$ that maximizes the second term on the right hand side in \ref{eq:ExpRev} and ensures that $K_i(b_i)=0$ for all $i$ and $b_i$ while satisfying the BIC and IR constraints would provide the largest expected profit. 

 In order to simplify  maximization of the second term in \ref{eq:ExpRev}  we assume that
  the virtual types $\Big( \theta_i - \frac{1 - F_{i}(\theta_i | b_i)}{f_{i}(\theta_i | b_i)}\Big)$ are non-decreasing in $\theta_i$ and $b_i$. Such a condition holds if $\frac{f_{i}(\theta_i | b_i)}{1 - F_{i}(\theta_i | b_i)}$ is non-decreasing in $\theta_i$ and $b_i$. This condition  can be viewed as a generalization of the increasing hazard rate condition \cite[Chapter 2]{borgers2015introduction} and is similar to the condition about monotonicity of virtual valuations  described in \cite{pai2013optimal} for multidimensional private types. We formally state this condition  and our assumptions below.

\textit{Generalized Monotone Hazard Rate Condition:} The type $(\theta_i, b_i)$ is said to be partially ordered above $(\theta_i', b_i')$, and this relation denoted by  $(\theta_i, b_i) \succeq (\theta_i', b_i')$, if $\theta_i \ge \theta_i'$ and $b_i \ge b_i'$. The distribution $f_{i}(\cdot, \cdot)$ then satisfies the generalized monotone hazard rate condition if:  \vspace{-.8em}
\begin{equation}\label{MHRC}
\begin{split}
 \hspace*{-0.3cm}(\theta_i, b_i) \succeq (\theta_i', b_i') \; \; \Longrightarrow \; \; \frac{f_{i}(\theta_i | b_i)}{1 - F_{i}(\theta_i | b_i)} \ge \frac{f_{i'}(\theta_i' | b_i')}{1 - F_{i'}(\theta_i' | b_i')};
\end{split}
\end{equation}
Further,  $b_i > b'_i$ and $\theta_i \geq \theta_i'$ imply
\begin{equation}\label{MHRC_2} \frac{f_{i}(\theta_i | b_i)}{1 - F_{i}(\theta_i | b_i)} > \frac{f_{i'}(\theta_i' | b_i')}{1 - F_{i'}(\theta_i' | b_i')}.
\end{equation}
   
\begin{assum}\label{assum:MHRC}
We assume that the probability density functions $f_i(\cdot, \cdot)$ satisfy the generalized monotone hazard rate condition for all $i \in \mathcal{N}$.
\end{assum}

\begin{assum} \label{assum:Negwmin}
We assume that $w_i(\theta_i^{min}, b_i) < 0 \; , \; \forall b_i \in \{1, 2, \cdots, k\} \; , \; \forall i \in \mathcal{N}$.
\end{assum}

The following theorem characterizes the optimal mechanism under the above assumptions.
\begin{theorem}\label{thm:2}
 Consider the allocation, purchase and tax functions $(\xi^*, g^*, t^*)$ defined below \vspace{-1em}
\begin{align} \label{qi_opt}
(\xi^*(\theta, b), g^*(\theta, b)) \in \argmax\limits_{ (\mathbf{a}, \mathbf{d}) \in \mathcal{S}(b)} \; & \Big\{ \sum\limits_{i=1}^N a_i w_i(\theta_i, b_i) - \sum\limits_{j=1}^k p_j d_j \Big\} \; , \vspace{-1em}
\end{align} 
where $a_i$ and $d_j$ are the $i$th and $j$th entries of the vectors $\mathbf{a}$ and $\mathbf{d}$ respectively; \vspace{-1em}
\begin{equation} \label{ti_opt}
t^*_i(\theta, b) := \theta_i \: \xi_i^{*}(\theta, b) - \: \int\limits^{\theta_i}_{\theta_i^{min}} \xi_i^{*}(s,\theta_{-i}, b) \: ds.
\end{equation}
Then, under Assumptions 1-4, $(\xi^*, g^*, t^*)$ is a profit-maximizing  Bayesian incentive compatible and individually rational mechanism.
\end{theorem}
\begin{proof} 
The proof is similar to the proof for Theorem 1 in \cite{navabi2016optimal}. 
\end{proof}
The optimal allocation and purchase vectors pair $(\xi^*(\theta, b), g^*(\theta,b))$ given in \eqref{qi_opt} is the solution of an integer program and hence computationally hard to obtain. Moreover, each type profile $(\theta, b) \in \Theta \times \{1, 2, \cdots, k\}^N$ requires the solution of a different integer program. Similarly, the characterization of payments given by \eqref{ti_opt} is not very useful from a computational viewpoint as it requires the solution of a continuum of integer programs. 

Navabi et al. \cite{navabi2016optimal} considered the same mechanism with fixed supply ($g_j = 0, j = 1, \cdots, k$) and showed that by leveraging the nested structure imposed on the consumers' flexibility sets (see \eqref{Bsets}) the optimal allocation as characterized in Theorem \ref{thm:2} can be simplified. We summarize their main results in the following subsection. \vspace{-.5em}
\subsection{Optimal Allocation for the Case with Fixed Supply} \label{sec:resOld}
In the mechanism $(\xi, g, t)$, suppose $g_j(\theta, b) = 0, j = 1, 2, \cdots, k$ for all $(\theta, b) \in \Theta \times \{1, 2, \cdots, k\}^N$, that is, the seller's available supply is fixed and  no further purchases will be made. Then, based on the results in \cite{navabi2016optimal} under the optimal allocation, consumer $l$ in class $\mathcal{C}_i$ gets a desired good if its virtual valuation exceeds 0 and the thresholds $w_i^{thr}, w_{i+1}^{thr}, \cdots, w_k^{thr}$ where, the thresholds $\{w_i^{thr}\}_{i=1}^k$ are obtained using an iterative algorithm constructed in \cite{navabi2016optimal}. Consider the non-negative integer quantities $r_1^*, \cdots, r_k^*$ that are obtained as the solutions to $k$ one-dimensional integer programs that are solved recursively (see Lemma 6 in \cite{navabi2016optimal}). The above thresholds are then constructed in terms of $r_1^*, \cdots, r_k^*$ through an iterative algorithm that we briefly outline here (for a complete description of this procedure see \cite[section V-B]{navabi2016optimal}):


\begin{enumerate}
\item Firstly, any consumer $l$ with $w_l(\theta_l, b_l) \leq 0$ is immediately removed from consideration and is not allocated any good.  For each class of consumers, define the subset of consumers who have positive virtual valuations: $\mathcal{C}_i^+ := \{l \in \mathcal{C}_i : w_l(\theta_l, b_l) >0  \}.$ 
Define $r^*_1,\ldots,r^*_k$ as in Lemma 6 of \cite{navabi2016optimal} accordingly.  
\item  Let $\mathcal{L}_1 := \mathcal{C}_1^+$. From $\mathcal{L}_1$, $r^*_1$ consumers with the lowest virtual valuations are removed from consideration\footnote{Ties are resolved randomly. For continuous valuations, ties happen with zero probability and therefore the allocation rule for ties does not affect expected profit.}. The set of remaining consumers in $\mathcal{L}_1$ is denoted by $\mathcal{N}_1$.
\item We now proceed iteratively:  For $2 \leq i \leq k$, given the set $\mathcal{N}_{i-1}$, define  $\mathcal{L}_i := \mathcal{N}_{i-1} \bigcup \mathcal{C}^+_i$. Remove $r^*_i$ consumers with lowest virtual valuations from $\mathcal{L}_i$. The set of remaining consumers in $\mathcal{L}_i$ is now defined as $\mathcal{N}_i$.
\item After the $k^{th}$ iteration, all consumers in $\mathcal{N}_k$ are allocated a good from their respective flexibility sets.
\end{enumerate}
The thresholds $\{w_i^{thr}\}_{i=1}^k$ are then defined as
\begin{equation} \label{ThetaThr}
w_i^{thr} := (r_i^*)^{\text{th}} \; \text{lowest virtual valuation in} \; \mathcal{L}_i , i =  1,2, \cdots, k.
\end{equation}
(If $r^*_i=0$, $w_i^{thr}=0$.)\\
Hence, under the optimal allocation,  consumer $l$ in class $\mathcal{C}_i$ gets a desired good if its virtual valuation exceeds $0$ and the thresholds $w_i^{thr}, w_{i+1}^{thr}, \cdots, w_k^{thr}$. 

\subsection{Optimal Purchase Decisions } \label{sec:PurchaseRule}
For each type profile $(\theta, b) \in \Theta \times \{1, 2, \cdots, k\}^N$, let $\mathcal{A}_i$ denote the set of consumers in class $\mathcal{C}_i$ who are served under the supply profile $(m_1, m_2, \cdots, m_k)$ through the allocation procedure outlined in Section \ref{sec:resOld}, that is, for $i = 1, 2, \cdots, k$ \vspace{-1em}
\begin{equation}\label{Ai}
\begin{split}
\mathcal{A}_i &\coloneqq \Big\{ l \in \mathcal{C}_i : w_l(\theta_l, i)  > \max\{ 0, w_i^{thr}, w_{i+1}^{thr}, \cdots, w_k^{thr} \} \Big\} \; , 
\end{split}
\end{equation} 
where, the thresholds $\{w_i^{thr}\}_{i=1}^k$ are given in \eqref{ThetaThr}. Define $\mathcal{A} = \bigcup\limits_{i=1}^k \mathcal{A}_i$. Considering the results pointed in Section \ref{sec:resOld}, $\mathcal{A}$ is thus the set of \textit{all} consumers who were served under the supply profile $(m_1, m_2, \cdots, m_k)$. 

The optimal purchase decisions are now characterized in the following lemma.
\begin{lemma}\label{lemma:g*}
For each type profile $(\theta, b) \in \Theta \times \{1, 2, \cdots, k\}^N$ define
\begin{equation}\label{E_i}
\mathcal{E}_i \coloneqq \{ j \in \mathcal{C}_i \setminus \mathcal{A}_i : w_j(\theta_j, i) > p_i  \} \; \; , \; \; i = 1, 2, \cdots, k. 
\end{equation}
Then, the optimal purchase decisions are
\begin{equation} \label{g*}
g^*_i(\theta, b) = |\mathcal{E}_i| \; , \; i = 1, 2, \cdots, k.
\end{equation}
\begin{proof}
See Appendix \ref{sec:g*_prf}.
\end{proof}
\end{lemma}

\subsection{Optimal Allocation } \label{sec:AllocRule}
With the optimal purchase decisions $g^* \coloneqq (g_1^*, g_2^*, \cdots, g_k^*)$ characterized in \eqref{g*}, the increased supply profile is $\bold{m}^{g^*} \coloneqq (m_1+g_1^*, m_2+g_2^*, \cdots, m_k+g_k^*)$. The problem is now similar to one with the fixed supply $\bold{m}^{g^*}$ as formulated and solved in \cite{navabi2016optimal}; hence, given the new supply profile $\bold{m}^{g^*}$ and using the results in Section \ref{sec:resOld} and Lemma \ref{lemma:g*}, it can be shown that under the optimal allocation, consumer $l$ in class $\mathcal{C}_i$ gets a desired good if its virtual valuation exceeds 0 and the thresholds $w_i^{thr}, w_{i+1}^{thr}, \cdots, w_k^{thr}$, \textit{or} the unit price $p_i$. Let us define
\begin{equation}\label{thet_thr_xi}
\begin{split}
\theta_{l,i}^{g^*} &:= \Big\{ x : w_l(x, i) = \min\big\{p_i , \max\{0, \{w_j^{thr}\}_{j=i}^k\} \big\} \Big\}, \\
& l \in \mathcal{C}_i, \; i = 1, 2, \cdots, k.
\raisetag{1\baselineskip}
\end{split}
\end{equation}
Because of the monotonicity of virtual valuation as a function of true valuation, consumer $l$ in class $\mathcal{C}_i$ gets a good if $\theta_l > \theta_{l,i}^{g^*}$. Thus,  \vspace{-.8em}
\begin{equation} \label{Alloc}
\begin{split}
\xi_l^{*}(\theta, b) = \left\{
    \begin{array}{ll}
        1 & \text{if} \minitab \theta_l > \theta_{l,i}^{g^*}  \\
        0 & \text{otherwise}  
    \end{array}
\right., \; \forall l \in \mathcal{C}_i, \; \;  i = 1, 2, \cdots, k .
\end{split}
\end{equation}
\vspace{-2em}
\subsection{Payment Functions} \label{sec:TaxFuncs}
We can now use the optimal allocation rule described in section \ref{sec:AllocRule} to simplify consumers' payment functions. 
From \eqref{ti_opt} the optimal payment function for consumer $l$ in flexibility class $\mathcal{C}_i$ has the following form:\vspace{-.8em} 
\begin{align} \label{eq:Tax_i}
t^*_l(\theta, b) = \theta_l \xi_l^{*}(\theta, b) - \int\limits^{\theta_l}_{\theta_l^{min}} \xi_l^{*}(s, \theta_{-l}, b) \: ds. 
\end{align}
Using the definition of  $\xi_l^{*}(\theta, b)$ given in \eqref{Alloc}, $t^*_l(\theta, b)$ can be simplified as:
\begin{enumerate}
\item If $\theta_l > \theta_{l,i}^{g^*}$, \vspace{-.8em}
\begin{equation} \label{taxFunc}
\begin{split}
\hspace{-2.6em}t^*_l(\theta, b) &= \theta_l  - \int\limits_{\theta_l^{min}}^{\theta_{l,i}^{g^*} } \underbrace{\xi_l^{*}(s, \theta_{-l}, b)}_{=0} \; ds - \int\limits_{\theta_{l,i}^{g^*} }^{\theta_l} \underbrace{\xi_l^{*}(s, \theta_{-l}, b)}_{=1} \; ds = \theta_{l,i}^{g^*}. 
\end{split}
\end{equation}
\vspace{-.8em}
\item If $\theta_l \leq  \theta_{l,i}^{g^*}$, $t^*_l(\theta,b) =0$.
\end{enumerate}

The optimal allocation and purchase decisions as well as the optimal payments can thus be  computed through the straightforward threshold-based procedure outlined in section \ref{sec:resOld}. By using the nested structure of  flexibility sets, this procedure obviates the need to solve the computationally hard integer program formulated in Theorem \ref{thm:2}.

\vspace{-.8em}
\section{Conclusion}\label{sec:recap}
We studied the problem of designing profit-maximizing mechanisms for allocating multiple goods to flexible consumers. We characterized the allocation and purchase rules for an incentive compatible, individually rational and profit-maximizing mechanism as the solution to an integer program. The  payment function was determined by the optimal allocation rule in the form of an integral equation. We then exploited the nested structure imposed on the flexibility sets to simplify the optimal mechanism and provided a complete characterization of allocations, purchase decisions and payments in terms of simple thresholds.

An interesting direction for further exploration is to study this mechanism problem in a dynamic framework where the population of the consumers that interact with the market as well as seller's supply undergo stochastic changes over time. Richer information structures may be needed to address those cases.
\vspace{-1em}
\input{gstar_prf-ext}

%

\bibliographystyle{IEEEtran}
\bibliography{REF}

\end{document}

%% file: intro.tex
\section{Introduction}\label{sec:intro}

Allocation of limited resources among a set of self-interested consumers whose preferences are their private information, arises in a number of applications including cognitive radio networks and communication networks. The allocation procedure would typically involve eliciting some information from the consumers. The self-interested consumers can be strategic in revealing their private information. Hence, there is a need to anticipate the consumers' strategic behavior in revelation of their private information. The economic theory of mechanism design provides a framework for studying these problems from a game theoretic perspective where the owner of the resources can design the rules of interactions among the consumers in a way that its desired objective emerges at the equilibrium outcome of the induced game. 


In this paper, we consider the problem of designing profit-maximizing mechanisms for allocating multiple heterogeneous goods to flexible consumers. 
We consider the flexibility model proposed by \cite{navabi2016optimal} where for each consumer a  \emph{flexibility set} is defined as the subset of goods the consumer finds equally desirable. Each consumer wants to consume \emph{one good from its flexibility set}. A consumer's flexibility set and the utility it gets from consuming a good from its flexibility set are both its private information. The flexibility sets considered in  \cite{navabi2016optimal} have a nested structure ---  each consumer's flexibility set can be one of the  $k$ sets, $\mathcal{B}_1, \mathcal{B}_2,\ldots, \mathcal{B}_k$, which are nested in the following way:\vspace{-1em}
\begin{equation}\label{nestB}
\mathcal{B}_1 \subset \mathcal{B}_2 \subset \cdots \subset \mathcal{B}_{k}.
\end{equation}

\subsection{Examples of Nested Flexibility}

There are several markets where consumer flexibility resembles the nested pattern in \eqref{nestB}.  An example comes from auction-based  spectrum allocation in cognitive radio networks (\cite{khaledi2013auction}, \cite{zhang2012auction}) where a primary spectrum owner has multiple frequency bands with different bandwidths. These bands  can be allocated to secondary users who need a certain minimum amount of bandwidth.  Suppose the primary owner has frequency bands of widths $w_1, w_2, \cdots, w_{k}$ with  $w_1 < w_2 < \cdots < w_k$. Let $\mathcal{W}_i, i = 1, 2, \cdots, k,$ denote the set of frequency bands of width $w_i$ that are available for allocation to secondary users.  Define $\mathcal{B}_i = \bigcup\limits_{j=k-i+1}^k \mathcal{W}_j, i = 1, 2, \cdots, k,$ as the set of frequency bands of width greater than or equal to $w_{k-i+1}$.  We thus have $\mathcal{B}_{1} \subset \mathcal{B}_{2} \subset \cdots \subset \mathcal{B}_{k}$. A secondary user that needs one frequency band of width at least $w_i$ can be interpreted as having  $\mathcal{B}_{k-i+1}$ as its flexibility set.


Consider next auction-based content delivery in Wireless Information Centric Networks \cite{mangili2016bandwidth} where multiple content providers compete for limited cache storage resources provided by a Wireless Access Point (WAP) in a given region for a certain time period. Suppose the WAP has $k$ cache servers with storage capacities $c_1 < c_2 <\ldots < c_k$. Assume that one cache server can serve at most one content provider at a time. Let $\mathcal{B}_i$ be the set of cache servers with capacity greater than or equal to $c_{k-i+1}$. Clearly the sets $\mathcal{B}_i, i=1,\ldots,k,$ are nested. A content provider who needs a  cache of storage capacity   at least $c_{k-i+1}$ has the flexibility set $\mathcal{B}_i$.



\subsection{Comparison with Prior Literature}\label{sec:prior-lit}

Numerous works have addressed social welfare maximizing or \textit{efficient} auctions, the most well-known of these being the Vickrey-Clarke-Groves (VCG) mechanism \cite{vickrey1961counterspeculation}, \cite{clarke1971multipart}, \cite{groves1973incentives}. Efficient auctions have also been extensively studied in the context of combinatorial auctions (\cite{cramton2006combinatorial}, \cite[Chapter 8]{young2014handbook}, \cite[Chapter 11]{nisan2007algorithmic}). Our focus in this paper, however, is on   revenue-maximizing auctions. 
In his seminal paper \cite{myerson1981optimal}, Myerson derived fundamental results for the single-unit revenue maximizing auction. In sequel, several works studied revenue-maximizing multi-unit auctions with identical goods under various assumptions about the consumers' utility model and private information structure.  The setups in  \cite{harris1981theory} and  \cite{maskin1989optimal}, for instance, include the problem of auctioning multiple identical goods among consumers with unit demand and private valuations. \cite{malakhov2009optimal} considered the auction of multiple identical goods to  consumers with limited capacities for the number of goods they can consume.
Unlike these models where all goods are perceived to be identical by all the consumers, in our model consumers differentiate between goods according to their flexibility sets.

The problem of designing revenue-maximizing auctions has also been investigated in the context of combinatorial auctions \cite[Section 5.2]{de2003combinatorial} where the seller has multiple heterogeneous items to auction and consumers can place bids on various \textit{combinations/bundles} of  goods.
Armstrong \cite{armstrong2000optimal}  and Avery et al. \cite{avery2000bundling} studied  revenue-maximizing auction for the case where the seller wants to sell two non-identical goods to several consumers  that can receive one or both of the goods. Unlike these setups, in our model each consumer wants at most one good and the number of goods can exceed two.

Some recent works (\cite{ledyard2007optimal}, \cite{abhishek2010revenue}) have studied the design of revenue-maximizing auctions when the consumers are single-minded. A single-minded consumer is interested in getting \textit{all} goods from a certain subset of goods. In contrast to these models, in our setup  each consumer wants to get \textit{one} good from its flexibility set.

\cite[Section 5.2]{de2003combinatorial} considered a general setup where the seller has multiple distinct goods and each consumer has a value function that gives its valuation for each subset of goods. For each allocation rule,  \cite{de2003combinatorial} provides a linear program whose solution (if it exists) gives a payment   rule  that satisfies incentive compatibility and individual rationality constraints. 
As mentioned in  \cite[Section 5.2]{de2003combinatorial}, this approach is computationally very demanding as the number of possible allocation rules can be very large and no closed-form solutions are available  in general. 
Our model can be viewed as a special case of the general framework of  \cite{de2003combinatorial}.  A consumer with flexibility set $\mathcal{B}_i$  and valuation $\alpha$ can be viewed as having a value function $v_i(\mathcal{S}) = \alpha$ if $\mathcal{S} \cap \mathcal{B}_i \neq \emptyset \mbox{~and~} |\mathcal{S}|=1$ and $0$ otherwise.
In sections \ref{sec:resOld}-\ref{sec:TaxFuncs}, we show that under the assumption of nested flexibility sets,  the optimal auction can be found in a computationally much simpler fashion.

%% file: gstar_prf-ext.tex
\section{Proof of Lemma 4.2}\label{sec:g*_prf} 
Consider the supply vector $(m_1+h_1, \cdots, m_k+h_k), h_1, \cdots, h_k \ge 0$. We want to show that $(g_1^*, \cdots, g_k^*)$ is the optimal purchase decision vector where, $g_i^* = |\mathcal{E}_i|, i = 1, 2, \cdots, k$ (see \eqref{E_i}). Let $\mathcal{A}_i(h_1, h_2, \cdots, h_k), i = 1, \cdots, k$ denote the set of consumers in class $\mathcal{C}_i$ who are served under the supply profile $(m_1+h_1, \cdots, m_k+h_k), h_1, \cdots, h_k \ge 0$. Recall that $\mathcal{A}_i, i = 1, 2, \cdots, k$ denotes the set of consumers in class $\mathcal{C}_i$ that are served under the supply profile $(m_1, m_2, \cdots, m_k)$ (see \eqref{Ai}) using the allocation procedure outlined in section \ref{sec:resOld}. 

We start from flexibility level $k$ and we consider the following two cases:
\newcounter{l1}
\newcommand{\barablistOne}{\begin{list}{($1.$\arabic{l1})}{\usecounter{l1}}}
\barablistOne 
\item $\mathcal{A}_k(h_1, h_2, \cdots, h_{k-1}, 0) \supset \mathcal{A}_k$, \label{item:Aksup}
\item $\mathcal{A}_k(h_1, h_2, \cdots, h_{k-1}, 0) = \mathcal{A}_k$.  \label{item:Akeq}
\end{list}
Case (1.\ref{item:Aksup}) implies that there exists some consumer in class $\mathcal{C}_k$ who is served with an additionally purchased good for some set $\mathcal{B}_i \subset \mathcal{B}_k$ at the price of $p_i$. In this case $(h_1, h_2, \cdots, h_{k-1}, h_k)$ cannot be optimal for \textit{any} $h_k \ge 0$; essentially because one can reduce $h_i$ by 1 and increase $h_k$ by 1 to improve the objective function value in \eqref{qi_opt} by the amount $(p_i - p_k)$. Hence, we consider purchase decision vectors $(h_1, h_2, \cdots, h_k) \; , \; h_1, \cdots, h_k \ge 0$ such that case (1.\ref{item:Akeq}) is true. 

Assumption \ref{assum:Nonz} implies that the $h_k$ purchased goods for $\mathcal{B}_k \setminus \mathcal{B}_{k-1}$ can only be used to serve the consumers in class $\mathcal{C}_k$. From case (1.\ref{item:Akeq}) it follows that $|\mathcal{A}_k(h_1, \cdots, h_{k-1}, h_k)|$ cannot exceed $|\mathcal{A}_k| + h_k$. Now consider the following two cases

\newcounter{l2}
\newcommand{\barablistTwo}{\begin{list}{($2.$\arabic{l2})}{\usecounter{l2}}}

\barablistTwo
\item $|\mathcal{A}_k(h_1, \cdots, h_{k-1}, h_k)| < |\mathcal{A}_k(h_1, h_2, \cdots, h_{k-1}, 0) | + h_k$, \label{item:hkle}
\item $|\mathcal{A}_k(h_1, \cdots, h_{k-1}, h_k)| = |\mathcal{A}_k(h_1, h_2, \cdots, h_{k-1}, 0) | + h_k$. \label{item:hkeq}
\end{list}

Case (2.\ref{item:hkle}) implies that at least one of the $h_k$ goods is not being allocated to any consumer in $\mathcal{C}_k$ and hence, is wasted. In this case also $(h_1, \cdots, h_{k-1}, h_k)$ cannot be optimal for any $h_k > 0$ since $h_k$ can be reduced by 1 to eliminate the incurred cost of the wasted good ($p_k$) and improve the value of the objective function; hence, we consider vectors $(h_1, \cdots, h_k)$ that satisfy case (2.\ref{item:hkeq}) where all the purchased $h_k$ goods are \textit{allocated} to the consumers in $\mathcal{C}_k$ and thus, are not wasted. 

Suppose now we start from $h_k = 0$ and increase it; as a result, the positive terms $(w_j - p_k), j \in \mathcal{E}_k$ (see \eqref{E_i}) are added in the objective function (see \eqref{qi_opt}) that increase its value. We continue to increase $h_k$ until $h_k = | \mathcal{E}_k | = g_k^*$; if we increase $h_k$ further such that $h_k > g_k^*$, it follows from case (2.\ref{item:hkeq}) --- all the $h_k$ goods must be allocated --- that the objective function value begins to decline; essentially because the terms $(w_j - p_k), j \in \mathcal{C}_k \setminus (\mathcal{E}_k \cup \mathcal{A}_k)$ that are added then, are not positive since $w_j < p_k \; , j \in \mathcal{C}_k \setminus (\mathcal{E}_k \cup \mathcal{A}_k)$; hence, $h_k = | \mathcal{E}_k | = g_k^*$ is the optimal number of goods to purchase at the price of $p_k$ and add to $\mathcal{B}_k \setminus \mathcal{B}_{k-1}$. 

The above arguments imply that in order for the purchase decision vector $(h_1, \cdots, h_k)$ to be optimal, it needs to satisfy the following:

\begin{enumerate}
\item Case (1.\ref{item:Akeq}) holds true, 
\item Case (2.\ref{item:hkeq}) holds true,
\item $h_k = | \mathcal{E}_k | = g_k^*$.
\end{enumerate}
Hence we restrict our attention to the decision vectors $(h_1, \cdots, h_k)$ that satisfy the above three conditions.

Inductive step: \\
Induction hypothesis: suppose now that the optimal purchase vector has the form $(h_1, \cdots, h_i, g_{i+1}^*, \cdots, g_k^*)$ where the $(i+1)$th up until $k$th entries are fixed as: $g_j^*= |\mathcal{E}_j|, j = i+1, \cdots, k$. 

We now want to show that in order for the decision vector $(h_1, \cdots, h_{i-1}, h_i, g^*_{i+1}, \cdots, g_k^*)$ to be optimal, it must be that $h_i = |\mathcal{E}_i| = g_i^*$. 

Consider the following cases 
\newcounter{l3}
\newcommand{\barablistThree}{\begin{list}{($3.$\arabic{l3})}{\usecounter{l3}}}
\barablistThree
\item For some $j \in \{ i+1, \cdots, k \}:$ \\
$\;\mathcal{A}_j(h_1, \cdots, h_{i-1}, h_i, g_{i+1}^*, \cdots, g_k^*) \supset \mathcal{A}_j(h_1, \cdots, h_{i-1}, 0, g_{i+1}^*, \cdots, g_k^*) $, \label{item:Ajsup}
\item $\mathcal{A}_j(h_1, \cdots, h_{i-1}, h_i, g_{i+1}^*, \cdots, g_k^*) = \mathcal{A}_j(h_1, \cdots, h_{i-1}, 0, g_{i+1}^*, \cdots, g_k^*)$, for all $j = i+1, \cdots, k$. \label{item:Ajeq}
\end{list}

Case (3.\ref{item:Ajsup}) implies that there is some consumer in some class $C_j, i+1 \le j \le k$ that is served with an additionally purchased good from the set $\mathcal{B}_i \setminus \mathcal{B}_{i-1}$ at the price of $p_i$;  clearly, in this case $(h_1, \cdots, h_{i-1}, h_i, g_{i+1}^*, \cdots, g_k^*)$ cannot be optimal for \textit{any} $h_i > 0$; simply because $h_i$ can be reduced by 1 while $h_j = g_j^*$ is increased to $h_j = g_j^*+1$ to enhance the objective function value in \eqref{qi_opt} by the amount $(p_i - p_j)$. Therefore, we consider decision vectors $(h_1, \cdots, h_{i-1}, h_i, g_{i+1}^*, \cdots, g_k^*)$ such that case (3.\ref{item:Ajeq}) holds true, that is, none of the consumers in classes $\mathcal{C}_j, j = i+1, \cdots, k$ are served with the purchased $h_i$ goods from the set $\mathcal{B}_i \setminus \mathcal{B}_{i-1}$.  

Now consider the following two cases:
\newcounter{l4}
\newcommand{\barablistFour}{\begin{list}{($4.$\arabic{l4})}{\usecounter{l4}}}
\barablistFour
\item $\mathcal{A}_i(h_1, \cdots, h_{i-1}, 0, g_{i+1}^*, \cdots, g_k^*) \supset \mathcal{A}_i$, \label{item:Aisup}
\item $\mathcal{A}_i(h_1, \cdots, h_{i-1}, 0, g_{i+1}^*, \cdots, g_k^*) = \mathcal{A}_i.$ \label{item:Aieq}
\end{list}

Case (4.\ref{item:Aisup}) implies that some consumer in class $\mathcal{C}_i$ is served with an additionally purchased good from some set $\mathcal{B}_l \setminus \mathcal{B}_{l-1}, l < i$; in this case also $(h_1, \cdots, h_{i-1}, h_i, g^*_{i+1}, \cdots, g_k^*)$ cannot be optimal for any $h_i \ge 0$; essentially because one can reduce $h_l$ by 1 and increase $h_i$ by 1 to improve the objective function value by $(p_l - p_i)$. Case (4.\ref{item:Aieq}) however indicates that none of the consumers in $\mathcal{C}_i$ are served with additionally purchased goods from the sets $\mathcal{B}_l\setminus \mathcal{B}_{l-1}, l < i$. We thus consider the decision vectors $(h_1, \cdots, h_{i-1}, h_i, g_{i+1}^*, \cdots, g_k^*)$ such that cases (3.\ref{item:Ajeq}) and (4.\ref{item:Aieq}) both hold true. 

Assumption \ref{assum:Nonz} and case (3.\ref{item:Ajeq}) imply that the additionally purchased $h_i$ goods can only be used to serve the consumers in $\mathcal{C}_i$. Also, from case (4.\ref{item:Aieq}) it follows that $|\mathcal{A}_i(h_1, \cdots, h_{i-1}, h_i, g_{i+1}^*, \cdots, g_k^*)|$ cannot exceed $|\mathcal{A}_i| + h_i$. Now consider the following two cases
\newcounter{l5}
\newcommand{\barablistFive}{\begin{list}{($5.$\arabic{l5})}{\usecounter{l5}}}
\barablistFive
\item $|\mathcal{A}_i(h_1, \cdots, h_{i-1}, h_i, g_{i+1}^*, \cdots, g_k^*)| < |\mathcal{A}_i(h_1, \cdots, h_{i-1}, 0, g_{i+1}^*, \cdots, g_k^*)| + h_i$, \label{item:hile}
\item $|\mathcal{A}_i(h_1, \cdots, h_{i-1}, h_i, g_{i+1}^*, \cdots, g_k^*)| = |\mathcal{A}_i(h_1, \cdots, h_{i-1}, 0, g_{i+1}^*, \cdots, g_k^*)| + h_i$. \label{item:hieq}
\end{list}

Case (5.\ref{item:hile}) implies that there is at least one of the $h_i$ goods that is not allocated to any consumer in $\mathcal{C}_i$ and hence, is wasted; it thus follows that in this case $(h_1, \cdots, h_{i-1}, h_i, g_{i+1}^*, \cdots, g_k^*)$ cannot be optimal for any $h_i > 0$ simply because one can reduce $h_i$ by 1 to eliminate the incurred cost of the wasted good and increase the value of the objective function by $p_i$. On the other hand, case (5.\ref{item:hieq}) implies that all the $h_i$ purchased goods are allocated to the consumers in $\mathcal{C}_i$; thus, we consider decision vectors $(h_1, \cdots, h_{i-1}, h_i, g_{i+1}^*, \cdots, g_k^*)$ that satisfy the cases (3.\ref{item:Ajeq}), (4.\ref{item:Aieq}) and (5.\ref{item:hieq}). Given this, suppose we start from $h_i = 0$ and increase its value; as a result, the positive terms $(w_j - p_i), j \in \mathcal{E}_i$ are added in the objective function and its values increases. We continue to increase $h_i$ until $h_i = | \mathcal{E}_i | = g_i^*$; if we increase $h_i$ further such that $h_i > g_i^*$, it follows from case (5.\ref{item:hieq}) that the terms $(w_j - p_i), j \in \mathcal{C}_i \setminus (\mathcal{E}_i \cup \mathcal{A}_i)$ that are added then, are not positive because $w_j < p_i \; , j \in \mathcal{C}_i \setminus (\mathcal{E}_i \cup \mathcal{A}_i)$; hence, $h_i = | \mathcal{E}_i | = g_i^*$ is the optimal number of goods to purchase at the price of $p_i$ and add to $\mathcal{B}_i \setminus \mathcal{B}_{i-1}$. 

From the above arguments it is deduced that in order for the decision vector $(h_1, \cdots, h_{i-1}, h_i, g_{i+1}^*, \cdots, g_k^*)$ to be optimal, it needs to satisfy the following:

\begin{enumerate}
\item Case (3.\ref{item:Ajeq}) holds true, 
\item Case (4.\ref{item:Aieq}) holds true, 
\item Case (5.\ref{item:hieq}) holds true,
\item $h_i = | \mathcal{E}_i | = g_i^*$.
\end{enumerate}

Therefore, we conclude inductively that the optimal purchase decision vector is $(g_1^*, \cdots, g_k^*)$, where $g_i^* = |\mathcal{E}_i|, i = 1, 2, \cdots, k$ gives the optimal number of goods to purchase at the price of $p_i$ for the set $\mathcal{B}_i \setminus \mathcal{B}_{i-1}$. This completes the proof.